\documentclass[conference,a4paper]{IEEEtran}
\usepackage[utf8]{inputenc}
\usepackage[T1]{fontenc}

\usepackage{cite}
\usepackage[cmex10]{amsmath}
\usepackage{array}
\usepackage{thmbox}
\newtheorem[L,rightmargin=6pt,thickness=1.5pt]{thmchapter}{Theorem}[section]

\usepackage{bm}
\usepackage{graphicx}
\usepackage{caption}
\usepackage{subcaption}
\usepackage{amssymb,ntheorem}
\usepackage[ruled,vlined]{algorithm2e}

\def\N{\mathbb{N}} 
\def\C{\mathbb{C}} 
\def\R{\mathbb{R}} 
\def\ell{\textit{l}} 
\def\wp1{\mathrm{w.p.} 1}  

\graphicspath{{figures/}}

\newtheorem{theoreme}{Theorem}[section]

\newtheorem*{lemme}{Useful lemmas.}

\newtheorem{proposition}[theoreme]{Proposition}

\newtheorem{definition}[theoreme]{Definition\rm}

\newtheorem*{remarque}{\it Remark}

\newtheorem*{intuition}{\it Intuition:\/}

\def\cc#1{\setlength{\tabcolsep}{0pt}\begin{tabular}{c}#1\end{tabular}}
\newcommand{\figc}[2][]
   {\setlength{\tabcolsep}{0pt}\begin{tabular}{c}\includegraphics[#1]{#2}\end{tabular}}

\def\yaxis#1{\cc{\rotatebox{90}{{\small #1}}}}

\usepackage{xcolor}

\usepackage[normalem]{ulem}

\definecolor{indigo}{rgb}{.29,0.,0.51}

\begin{document}


\title{Travelling salesman-based variable density sampling}

\author{\IEEEauthorblockN{Nicolas Chauffert, Philippe Ciuciu}
\IEEEauthorblockA{CEA, NeuroSpin center, \\
INRIA Saclay, PARIETAL Team \\
145, F-91191 Gif-sur-Yvette, France\\
Email: firstname.lastname@cea.fr}
\and
\IEEEauthorblockN{Jonas Kahn}
\IEEEauthorblockA{Laboratoire Painlev\'e, UMR 8524\\
    Universit\'e de Lille 1, CNRS\\
 Cit\'e Scientifique B\^at. M2\\
 59655 Villeneuve d'Ascq Cedex, France\\
Email: jonas.kahn@math.univ-lille1.fr}
\and
\IEEEauthorblockN{Pierre Weiss}
\IEEEauthorblockA{ITAV, USR 3505\\
PRIMO Team, \\
Universit\'e de Toulouse, CNRS\\
Toulouse, France\\
Email: pierre.weiss@itav-recherche.fr}}
\maketitle

\begin{abstract}
Compressed sensing theory indicates that selecting a few measurements independently at random is a near optimal
strategy to sense sparse or compressible signals. 
This is infeasible in practice for many acquisition devices that acquire samples along \textit{continuous} trajectories. 
Examples include magnetic resonance imaging (MRI), radio-interferometry, mobile-robot sampling, ...
In this paper, we propose to generate continuous sampling trajectories by drawing a small set of measurements
independently and joining them using a travelling salesman problem solver.
Our contribution lies in the theoretical derivation of the appropriate probability density of the initial drawings.
Preliminary simulation results show that this strategy is as efficient as independent drawings while being
implementable on real acquisition systems.
\end{abstract}

\IEEEpeerreviewmaketitle

\section{Introduction}

Compressed sensing theory provides guarantees on the reconstruction quality of sparse and compressible signals $x\in \R^n$ from a limited number of linear measurements $(\langle a_k, x\rangle)_{k\in K}$.
In most applications, the measurement or acquisition basis $A=(a_k)_{k\in \{1,\cdots,n\}}$ is fixed (e.g. Fourier or Wavelet basis). 
In order to reduce the acquisition time, one then needs to find a set $K$ of minimal cardinality 
that provides satisfactory reconstuction results.
It is proved in~\cite{candes2006near,rauhut2010compressive} that a good way to proceed consists of drawing the indices of $K$ independently at random according to a distribution $\tilde \pi$ that depends on the sensing basis $A$.
This result motivated a lot of authors to propose variable density random sampling strategies (see e.g.~\cite{lustig2007sparse,knoll2011adapted,puy2011variable,krahmer2012beyond,chauffert2013}).
Fig.~\ref{fig:Distributions}(a) illustrates a typical sampling pattern used in the MRI context. 
Simulations confirm that such schemes are efficient in practice. 
Unfortunately, they can hardly be implemented on real hardware where the physics of the acquisition processes imposes \emph{at least} continuity of the sampling trajectory and sometimes a higher level of smoothness. Hence, actual CS-MRI solutions relie on adhoc
solutions such as random radial or randomly perturbed spiral trajectories to impose gradient continuity. 
Nevertheless these strategies strongly deviate from the theoretical setting and experiments confirm their practical suboptimality.

In this work, we propose an alternative to the independent sampling scheme. 
It consists of picking a few samples independently at random according to a distribution $\pi$ and joining them using a travelling salesman problem~(TSP) solver in order to design continuous trajectories. The main theoretical result of this paper states that $\pi$ should be proportional to $\tilde \pi^{d/(d-1)}$ where $d$ denotes the space dimension~(e.g. $d=2$ for 2D images, $d=3$ for 3D images) in order to emulate an independent drawing from distribution $\tilde \pi$. Similar ideas were previously proposed in the literature~\cite{wang2012smoothed}, but it seems that no author made this central observation. 

The rest of this paper is organized as follows. The notation and definitions are introduced in Section~\ref{notations}. Section~\ref{main} contains the main result of the paper along with its proof. Section~\ref{algo} shows how the proposed theory can be implemented in practice. Finally, Section~\ref{sec:results} presents simulation results in the MRI context. 
 
\section{Notation and definitions}
\label{notations}

We shall work on the hypercube $\Omega = [0,1]^d$ with $d \geq 2$. Let $m\in \N$. The set $\Omega$ will be partitionned in $m^d$ congruent hypercubes $(\omega_i)_{i\in I}$ of edge length $1/m$.
In what follows, $\left\{ x_i \right\} _{i\in \mathbb{N}^* }$ denotes a sequence of points in the hypercube $\Omega$, independently drawn from a density $\pi:\Omega \mapsto \R_+$. The set of the first $N$ points is denoted $X_N = \left\{ x_i \right\} _{i\leqslant N}$. 
For a set of points $F$, we consider the solution to the TSP, that is the shortest Hamiltonian path between those points.
We denote $T(F)$ its length. For any set $R\subseteq \Omega $ we define $T(F, R) = T(F \cap R)$.

We also introduce $C(X_N, \Omega )$ for the optimal curve itself, and $\gamma _N: [0,1] \to \Omega $ the function that parameterizes $C(X_N, \Omega )$ by moving along it at constant speed $T(X_N, \Omega )$. 

The Lebesgue measure on an interval $[0,1]$ is denoted $\lambda _{[0,1]}$. We define the \textit{distribution of the TSP solution as follows}.


\begin{definition}
    The distribution of the TSP solution is denoted $\tilde{\Pi}_N$ and defined, for any Borelian $B$ in $\Omega $ by:
    \begin{align*}
        \tilde{\Pi}_N(B) & = \lambda _{[0,1]} \left( \gamma _N^{-1} (B) \right).
    \end{align*}
\end{definition}

\begin{remarque}
    The distribution $\tilde{\Pi}_N$ is defined for fixed $X_N$. It makes no reference to the stochastic component of $X_N$.
\end{remarque}

In order to prove the main result, we need to introduce other tools. 
For a subset $\omega_i \subseteq \Omega $, we denote the length of $C(X_N, \Omega )\cap \omega_i$ as $T_{|\omega _i}(X_N, \Omega )=T(X_N, \Omega ) \tilde{\Pi}_N(\omega _i)$. Using this definition, it follows that:
\begin{equation}
\label{eq:defalternative}
\tilde{\Pi}_N(B) = \frac{T_{|B}(X_N, \Omega )}{T(X_N,\Omega)}, \ \forall B. 
\end{equation}

Let $T_B(F, R)$ be the length of the boundary TSP on the set $F \cap R$. 
The boundary TSP is defined as the shortest Hamiltonian tour on $F \cap R$ for the metric obtained from the Euclidean metric by the quotient of the boundary of $R$, that is $d(a,b) = 0$ if $a, b \in \partial R$. Informally, it matches the original TSP while being allowed to travel along the boundary for free. We refer to~\cite{yukich_gutin2002traveling} for a complete description of this concept.

\section{Main theorem}
\label{main}

Our main theoretical result reads as follows.

\begin{theoreme}
    \label{convergence_proba}
Define the density $\tilde \pi = \frac{\pi^{(d-1)/d}}{\int_{\Omega} \pi^{(d-1)/d}(x) dx}$.  Then almost surely with respect to  the law $\pi^{\otimes \mathbb{N} }$ of the sequence $\{x_i\}_{i\in \mathbb{N} ^*}$ of random points in the hypercube, the distribution $\tilde{\Pi}_N$ converges in distribution to $\tilde \pi$:
    \begin{align}
        \label{convForm}
        \tilde \Pi_N & \stackrel{(d)}{\rightarrow} \tilde \pi & \mbox{$\pi^{\otimes \mathbb{N}}$-a.s.}
    \end{align}
\end{theoreme}







\begin{intuition}
Let us first provide a rough intuition of the result since the exact proof is technical. 
The distribution $\tilde{\Pi}_N$ in a small cube is the relative length of the TSP in this cube. The number of points $N_c$ in the cube is proportional to $\pi$. Approximately, the TSP connects the points with other points in the cube, typically their neighbours, since they are close. Now, the typical distance between two neighbours in the cube is proportional to $N_c^{-1/d}$ or $\pi^{-1/d}$. So that the total length of the TSP in the small cube is proportional to $\pi\pi^{-1/d} = \pi^{1 - 1/d} \propto \tilde{\pi}$.
\end{intuition}

The remainder of this section is dedicated to proving this result.
The following proposition is central to obtain the proof:

\begin{proposition}
    \label{cube_by_cube}
    Almost surely, for all $\omega_i$ in $\{\omega_i\}_{1\leq i \leq m^d}$:    
    \begin{align}
        \label{for_small}
        \lim_{N\to \infty} \tilde{\Pi}_N(\omega _i) & = \tilde{\pi}(\omega _i) \\
                                                   & = \frac{\int_{\omega _i} \pi^{(d-1)/d}(x) \mathrm{d}x}{\int_{\Omega} \pi^{(d-1)/d}(x) \mathrm{d}x} & \mbox{$\pi^{\otimes \mathbb{N}}$-a.s.}
    \end{align}
   
\end{proposition}

The strategy consists in proving that $T_{|\omega _i}(X_N, \Omega )$ tends asymptotically to $T(X_N, \omega _i)$. The estimation of each term can then be obtained by applying the asymptotic result of Beardwood, Halton and Hammersley~\cite{beardwood1959shortest}:
\begin{theoreme}
    \label{BHH}
    If $R$ is a Lebesgue-measurable set in $\mathbb{R}^d$ such that the boundary $\partial R$ has zero measure, and $\{y_i\}_{i\in \mathbb{N}^* }$, with $Y_N = \left\{ y_i \right\} _{i\leqslant N}$ is a sequence of i.i.d. points from a density $p$ supported on $R$, then, almost surely,
    \begin{align}
        \label{BHHeq}
        \lim_{N \to \infty} \frac{T(Y_N, R) }{N^{(d-1)/d}} & = \beta (d) \int_{R}  p^{(d-1)/d}(x) \mathrm{d}x,
    \end{align}
    where $\beta (d)$ depends on the dimension $d$ only.
\end{theoreme}
We shall use a set of classical results on TSP and boundary TSP, that may be found in the survey books~\cite{yukich_gutin2002traveling}
and~\cite{yukich1998probability}.
\begin{lemme}
Let $F$ denote a set of $n$ points in $\Omega $.

\begin{enumerate}
 \item The boundary TSP is superadditive, that is, if $R_1$ and $R_2$ have disjoint interiors.
\begin{align}
    \label{superadditive}
    T_B(F, R_1 \cup R_2) & \geqslant T_B(F, R_1) + T_B(F, R_2).
\end{align}
 \item The boundary TSP is a lower bound on the TSP, both globally and on subsets. If $R_2 \subset R_1$:
\begin{align}
    \label{boundT}
    T(F, R) & \geqslant T_B(F, R)    \\
    \label{boundlocal}    
    T_{|R_2}(F, R_1) & \geqslant T_B(F, R_2)
\end{align}
  \item The boundary TSP approximates well the TSP~\cite[Lemma $3.7$]{yukich1998probability}):

\begin{align}
    \label{bound_approx}
    |T(F,\Omega) - T_B(F,\Omega)| =  O(n^{(d-2)/(d-1)} ).
\end{align}
  \item The TSP in $\Omega$ is well-approximated by the sum of TSPs in a grid of $m^d$ congruent hypercubes~\cite[Eq.~(33)]{yukich_gutin2002traveling}.
\begin{align}
    \label{approx}
    \lvert T(F, \Omega ) - \sum_{i=1}^{m^d} T(F, \omega _i) \rvert = O(n^{(d-2)/(d-1)}).
\end{align}
\end{enumerate}
\end{lemme}
We now have all the ingredients to prove the main results.
\begin{proof}[Proof of Proposition \ref{cube_by_cube}]

\begin{align*}
    \sum_{i\in I} T_B(X_N,\omega_i) & \stackrel{\eqref{superadditive}}{\leqslant} T_B(X_N,\Omega)  \\
  & \stackrel{\eqref{boundT}}{\leqslant} T(X_N,\Omega) = \sum_{i\in I} T|_{\omega_i}(X_N,\Omega) \\
  & \stackrel{\eqref{approx}}{\leqslant} \sum_{i\in I} T(X_N,\omega_i) + O(N^{(d-1)/(d-2)})
\end{align*}
Let $N_i$ be the number of points of $X_N$ in $\omega _i$. 

Since $N_i \leqslant N$, we may use the bound \eqref{bound_approx} to get: 
\begin{equation}
\label{eq:9}
 \lim_{N\rightarrow \infty}\frac{T(X_N,\omega_i)}{N^{(d-1)/d}} =  \lim_{N\rightarrow \infty}\frac{T_B(X_N,\omega_i)}{N^{(d-1)/d}}.
\end{equation}
Using the fact that there are only finitely many $\omega _i$, the following equalities hold almost surely:
\begin{align*}
\lim_{N\rightarrow \infty}  \frac{\sum_{i\in I}T_B(X_N,\omega_i)}{N^{(d-1)/d}} 
     & = \lim_{N\rightarrow \infty} \frac{\sum_{i\in I}T(X_N,\omega_i)}{N^{(d-1)/d}} \\
     & \stackrel{\eqref{approx}}{=} \lim_{N\rightarrow \infty}  \frac{\sum_{i \in I}T_{|\omega_i}(X_N,\Omega )}{N^{(d-1)/d}}.
\end{align*}

Since the boundary TSP is a lower bound~(cf. Eqs.~\eqref{boundlocal}-\eqref{boundT}) to both local and global TSPs, the above equality ensures that:
\begin{align}
    \label{allequal}
    \lim_{N\rightarrow \infty}  \frac{T_B(X_N,\omega_i)}{N^{(d-1)/d}} & = \lim_{N\rightarrow \infty} \frac {T(X_N,\omega_i )}{N^{(d-1)/d}}\\
     & = \lim_{N\rightarrow \infty} \frac {T_{|\omega_i}(X_N,\Omega )}{N^{(d-1)/d}}   & \mbox{$\pi^{\otimes \mathbb{N} }$-a.s, $\forall i$}.\nonumber
\end{align}
Finally, by the law of large numbers, almost surely $N_i / N \to \pi(\omega _i)=\int_{\omega_i} \pi(x)dx$. 
The law of any point $x_j$ conditioned on being in $\omega _i$ has density $\pi / \pi(\omega_i)$. By applying Theorem \ref{BHH} to the hypercubes $\omega _i$ and $\Omega$ we thus get:
\begin{align*}
    \lim_{N\rightarrow +\infty} \frac{T(X_N,\omega_i)}{N^{(d-1)/d}} & = \beta(d) \int_{\omega_i} \pi(x)^{(d-1)/d}dx & \mbox{$\pi^{\otimes \mathbb{N} }$-a.s, $\forall i$}.
\end{align*}
and 
\begin{align*}
    \lim_{N\rightarrow +\infty} \frac{T(X_N,\Omega)}{N^{(d-1)/d}} & = \beta(d) \int_{\Omega} \pi(x)^{(d-1)/d}dx & \mbox{$\pi^{\otimes \mathbb{N} }$-a.s, $\forall i$}.
\end{align*}
Combining this result with Eqs.~\eqref{allequal} and \eqref{eq:defalternative} yields Proposition~\ref{cube_by_cube}.
\end{proof}

\begin{proof}[Proof of Theorem \ref{convergence_proba}]
    Let $\varepsilon > 0$ and $m$ be an integer such that $\sqrt{d} m^{-d} < \varepsilon$. Then any two points in $\omega _i$ are at distance less than $\varepsilon $.
    
    Using Theorem \ref{cube_by_cube} and the fact that there is a finite number of $\omega _i$, almost surely, we get:
        $\lim_{N\rightarrow +\infty} \sum_{i\in I} \left| \tilde{\Pi}_N(\omega_i) - \tilde \pi(\omega _i) \right| = 0$. 
 Hence, for any $N$ large enough, there is a coupling $K$ of $\tilde{\Pi}_N$ and $\tilde \pi$ such that both corresponding random variables are in the same $\omega _i$ with probability $1- \varepsilon $. 
Let $A\subseteq \Omega$ be a Borelian. The coupling satisfies $\tilde{\Pi}_N(A) = K(A \otimes \Omega)$ and $\tilde \pi(A) = K(\Omega \otimes A)$. Define the $\varepsilon$-neighborhood by $A^\varepsilon=\{X\in \Omega \, | \, \exists Y \in A, \  \|X-Y\|<\varepsilon \}$. Then, we have:
$\tilde{\Pi}_N(A) =K(A \otimes \Omega)=K(\{A \otimes \Omega\} \cap \{ |X - Y| < \varepsilon\})+ K(\{A \otimes \Omega\} \cap \{ |X - Y| \geqslant \varepsilon \})$. It follows that:
\begin{align*}
       \tilde{\Pi}_N(A)&\leqslant  K({A \otimes A^{\epsilon}}) + K(|X - Y| \geqslant \varepsilon) \\
       &\leqslant  K(\Omega \otimes A^{\varepsilon})  + \varepsilon 
=    \tilde \pi(A^{\varepsilon}) + \varepsilon.
\end{align*}

This exactly matches the definition of convergence in the Prokhorov metric, which implies convergence in distribution.
\end{proof}
\section{Algorithm}
\label{algo}
The results presented in the previous section can be used to design a continuous sampling pattern with a target density $\tilde \pi$. 
The following algorithm summarizes this idea.
\begin{algorithm}

  \KwIn{$\tilde \pi:\Omega\mapsto \R_+$: a target sampling density.\\
  			$N$: an initial number of drawings.\\}
  \KwOut{A continuous sampling curve $C$.}
  \Begin{Define $\pi = \frac{\tilde \pi^{d/(d-1)}}{\int_{\Omega} \tilde \pi^{d/(d-1)}(x) dx}$.\\
  			 Draw $N$ points independently at random according to density $\pi$. \\
  			 Link these points with a travelling salesman solver to generate the curve $C$.
      }
\caption{An algorithm to generate a continuous sampling pattern according to a target density. \label{algo:algo}}
\end{algorithm}

Applying this algorithm raises various questions: how to choose the target density $\tilde \pi$? How to set the initial number of points $N$? Can the travelling salesman problem be solved for millions of points? We give various hints to the previous questions below.

\paragraph{Choosing a density $\tilde \pi$} We believe that this question is still treated superficially in the literature and deserves attention. Various strategies can be considered. A common empirical method consists in learning a density on image databases \cite{knoll2011adapted}. In the cases of Fourier measurements, this leads to the use of polynomially decreasing densities from low to high frequencies. 
The same strategy was proposed in \cite{lustig2007sparse} with no theoretical justification. The compressed sensing results allow to derive mathematically founded densities \cite{rauhut2010compressive,puy2011variable}. However, as outlined in \cite{chauffert2013}, an important ingredient is missing for these theories to provide good reconstruction results. The standard CS theory relies on the hypothesis that the signal is sparse, with no assumption on the sparsity structure. This makes the current theoretically founded sampling strategies highly sub-optimal. Recent works partially address this problem (see e.g. the review paper \cite{duarte2011structured}). However, to the best of our knowledge, the recent focus is on modifying the reconstruction algorithm, rather than deriving optimal sampling patterns.

\paragraph{Choosing an initial number of points $N$}

In applications, one usually wishes to sample $\tilde N$ points out of the $n$ possible ones. One should thus choose $N$ so that the discretized TSP trajectory contains $\tilde N$ points. This problem is well studied in the TSP literature \cite{beardwood1959shortest,rhee1989sharp}. Theorem \ref{BHH} ensures that the length of the TSP trajectory obtained by drawing $N$ points should be close to $L(N)=N^{(d-1)/d} \beta (d) \int_{R}  p^{(d-1)/d}(x) \mathrm{d}x$ where $\beta(d)$ can be evaluated numerically. Concentration results by Talagrand \cite{rhee1989sharp} show that this approximation is very accurate for moderate to large values of $N$. In order to obtain a discrete set of measurements from the continuous trajectory generated by Algorithm \ref{algo:algo}, we may discretize it with a stepsize $\Delta t$. The total number of points sampled is thus $N_s\simeq \lfloor \frac{L(N)}{\Delta t} \rfloor$ if an arclength parameterization is used. A possible way of obtaining approximately $\tilde N$ samples is thus to set:
\begin{equation}
N= \lfloor \Delta t L^{-1}(\tilde N)\rfloor.
\end{equation}
\paragraph{Solving the TSP}

Designing algorithms to solve the TSP is a widely studied problem. The book \cite{yukich_gutin2002traveling} provides a comprehensive review of exact and approximate algorithms. The TSP is known to be NP-hard and we cannot expect to solve it exactly for a large number of points $N$. From a theoretical point of view, Arora \cite{arora1998polynomial} shows that the TSP solution can be approximated to a factor $(1+\epsilon)$ with a complexity $O(N\log(N)^{1/\epsilon})$. 
From a practical point of view, there exist many heuristic algorithms that perform well in practice. The heuristics range from those that get within a few percent of optimum for 100,000-city instances in seconds to those that get within fractions of a percent of optimum for instances of this size in a few hours. In our experiments, we used a genetic algorithm~\cite{merz97}.

\section{Simulation results in MRI}
\label{sec:results}

The proposed sampling algorithm was assessed in a 2D MRI acquisition setup where images are sampled in the 2D Fourier
domain and compressible in the wavelet domain. Hence, $A=\mathcal{F}^*\Psi$ where $\mathcal{F}^*$ and $\Psi$ denote the discrete Fourier and inverse discrete wavelet transform, respectively.
Following~\cite{chauffert2013}, it can be shown that a near optimal sampling strategy consists of probing  $m$ independent samples of the 2D Fourier plane $(k_x,k_y)$ drawn independently from a target density $\tilde \pi$. The image is then reconstructed by solving the following $\ell^1$ problem using a Douglas-Rachford algorithm:
\begin{equation*}
 x^*=\mathop{\mathrm{argmin}}_{A_m x=y} \|x\|_1
\end{equation*}
where $A_m\in \C^{m\times n}$ is the sensing matrix, $x^*\in \C^n$ is the reconstructed image and $y\in \C^n$ is the collected data.
A typical realization is illustrated in Fig.~\ref{fig:Distributions}(a) which in practice cannot be implemented since MRI requires probing samples along continuous curves.
To circumvent such difficulties, a TSP solver was applied to such realization in order to join all samples through a countinuous trajectory, as illustrated in Fig.~\ref{fig:Distributions}(c). Finally, Fig.~\ref{fig:Distributions}(e) shows a curve generated by a TSP solver after drawing the same amount of Fourier samples from the density $\tilde \pi^2$ as underlied by Theorem~\ref{convergence_proba}. In all sampling schemes the number of probed Fourier coefficients was equal to one fifth of the total number~(acceleration factor $r=5$).

Figs.~\ref{fig:Distributions}(b,d,f) show the corresponding reconstruction results. It is readily seen that an independent random drawing from $\tilde \pi^2$ followed by a TSP-based solver yields promising results. Moreover, a dramatic improvement of $10$dB was obtained compared to the initial drawing from $\tilde \pi$.

\begin{figure}
\begin{center}
\begin{tabular}{cc}
\yaxis{$k_y$} \figc[width=.22\textwidth]{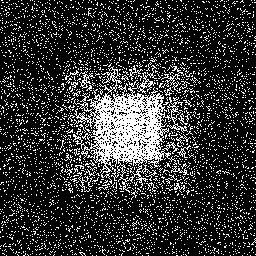}& \figc[width=.22\textwidth]{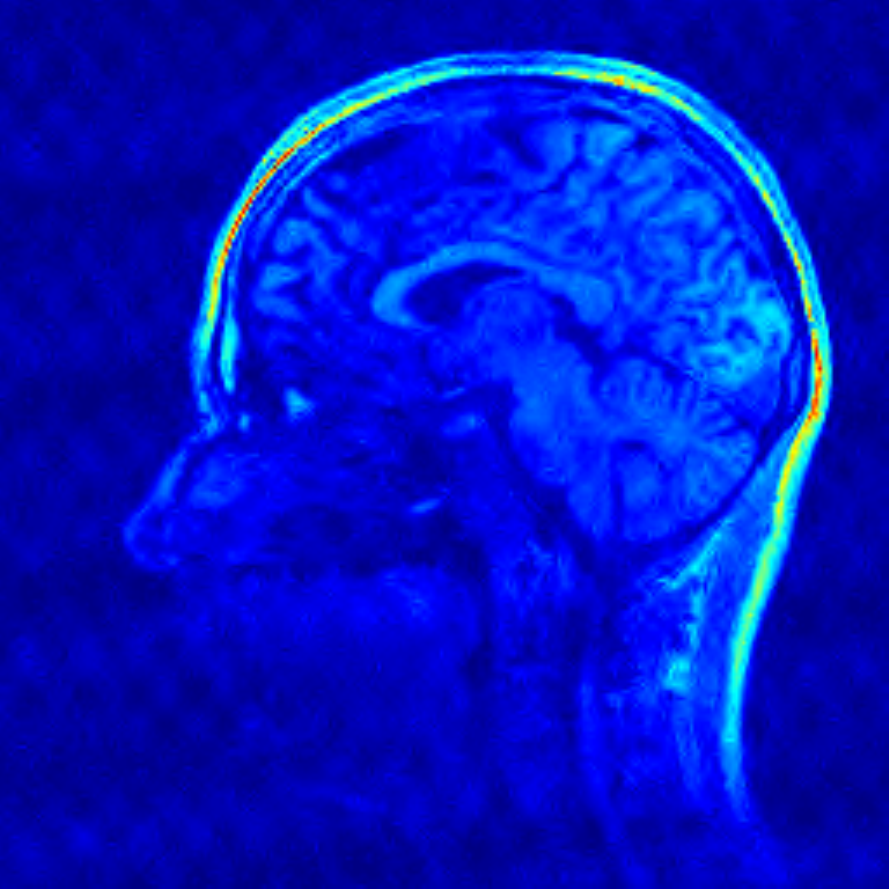}\\[-4.5cm]
{\small (a)} & {\small (b) SNR=33.0dB}\\[4.4cm]
\yaxis{$k_y$} \figc[width=.22\textwidth]{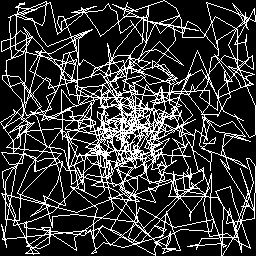}& \figc[width=.22\textwidth]{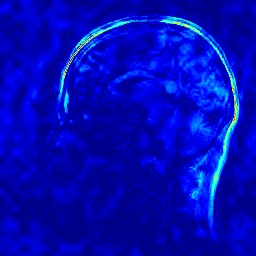}\\[-4.5cm]
{\small (c)} & {\small (d) SNR=24.1dB}\\[4.4cm]
\yaxis{$k_y$} \figc[width=.22\textwidth]{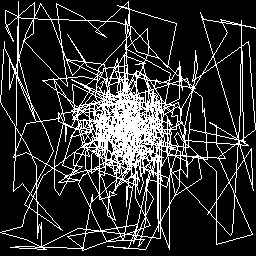} & \figc[width=.22\textwidth]{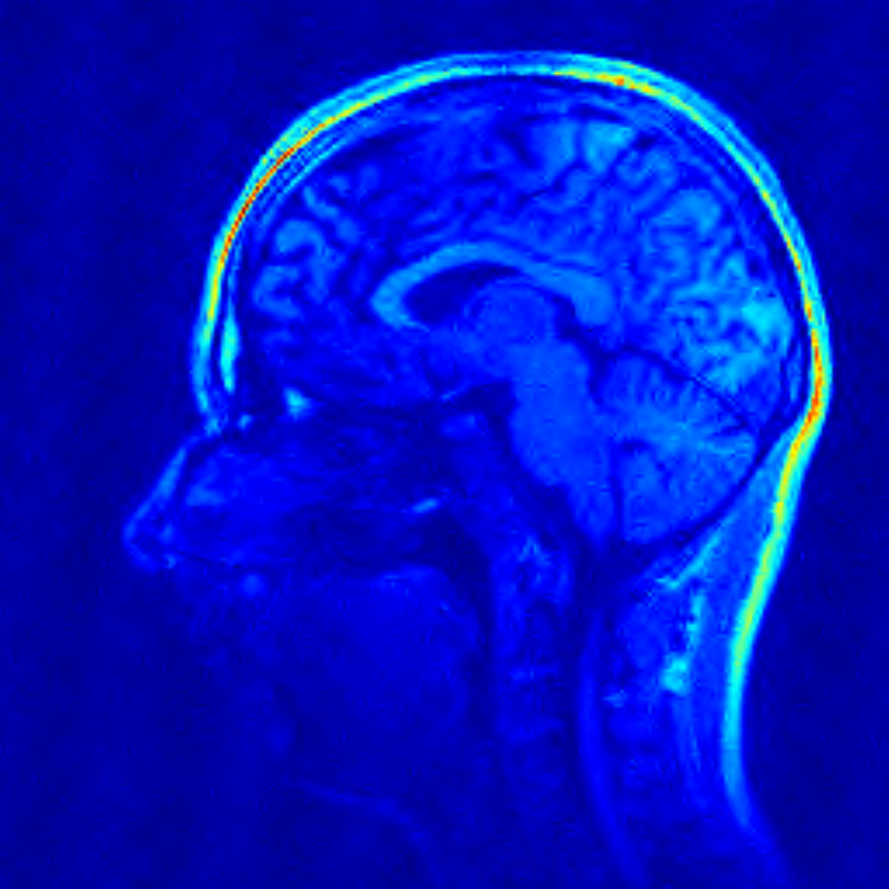}\\ 
 $k_x$ & \\[-4.9cm]
{\small (e)} & {\small (f) SNR=34.1dB}\\[4.6cm]
\end{tabular}\vspace*{-.5cm}
\end{center}
\caption{{\bf Left:} different sampling patterns (with an acceleration factor $r=5$). {\bf Right:} reconstruction results. From top to bottom: independent drawing from distribution $\tilde \pi$~(a), the same followed by a TSP solver~(c) and finally independent drawing
from distribution $\tilde \pi^2$ followed by a TSP solver.\vspace*{-.5cm}
\label{fig:Distributions}}
\end{figure}

\section{Conclusion}
Designing sampling patterns lying on continuous curves is central for practical applications such as MRI.
In this paper, we proposed and justified an original two-step approach based on a TSP solver to produce such continuous trajectories. 
It allows to emulate any variable density sampling strategy and could thus be used in a large variety of applications.
In the above mentioned MRI example, this method improves the signal-to-noise ratio by $10$dB compared to more naive approaches and provides results similar to those obtained using unconstrained sampling schemes.
From a theoretical point of view, we plan to assess the convergence rate of the empirical law of the  travelling salesman trajectory to the target distribution $\pi^{(d-1)/d}$. From a practical point of view, we plan to develop 
algorithms that integrate stronger constraints into account such as the maximal curvature of the sampling trajectory, which plays a key role in many applications. 

\section*{Acknowledgment}


The authors would like to thank the mission pour l'interdisciplinarit\'e from CNRS and the ANR SPHIM3D for partial support of Jonas Kahn's visit to Toulouse and the CIMI Excellence Laboratory for inviting Philippe Ciuciu on an excellence researcher position during winter 2013.




\footnotesize{
\bibliographystyle{IEEEbib}
\bibliography{bibSampTA_VC}
}
%
%


\end{document}